\documentclass[11pt]{article}%

\usepackage{amsthm,amsmath,amsfonts}
\usepackage{amssymb, url}
\usepackage{graphicx}
\usepackage{amsmath}
\usepackage{amsfonts}
\usepackage{amssymb}%

\usepackage{enumerate}
\usepackage{enumitem}
\setlist{topsep=2pt,itemsep=0pt,partopsep=4pt, parsep=4pt}

\theoremstyle{plain}

\newtheorem{lemma}{Lemma}
\newtheorem{proposition}{Proposition}

\theoremstyle{definition}

\theoremstyle{remark}

\newcommand{\E}{\mathbf{E}}
\newcommand{\p}{\mathbf{P}}
\newcommand{\R}{\mathbb{R}}

\newcommand{\Var}{\mathbf{Var}}

\newcommand{\bM}{\mathbf{M}}

\newcommand{\bX}{\mathbf{X}}

\newcommand{\bc}{\mathbf{c}}

\begin{document}

\title{Stochastic modeling of  in vitro bactericidal potency}

\author{
Anita Bogdanov\thanks{
Department of Medical Microbiology and 
Immunobiology, University of Szeged,
D\'om t\'er 10, 6720 Szeged, Hungary;
e-mail: \texttt{varga-bogdanov.anita@med.u-szeged.hu}
},
P\'eter Kevei\thanks{
Bolyai Institute, University of Szeged, 
Aradi v\'ertan\'uk tere 1, 6720 Szeged, Hungary; 
e-mail: \texttt{kevei@math.u-szeged.hu}}, 
M\'at\'e Szalai\thanks{
Bolyai Institute, University of Szeged, 
Aradi v\'ertan\'uk tere 1, 6720 Szeged, Hungary; 
e-mail: \texttt{szalaim@math.u-szeged.hu}},
Dezs\H{o} Virok\thanks{
Department of Medical Microbiology and 
Immunobiology, University of Szeged,
D\'om t\'er 10, 6720 Szeged, Hungary;
e-mail: \texttt{virok.dezso.peter@med.u-szeged.hu}}
}

\date{}

\maketitle

\begin{abstract}
We provide a Galton--Watson model for the growth of a 
bacterial population in the presence of antibiotics.
We assume that bacterial cells either die or 
duplicate, and the corresponding probabilities 
depend on the concentration of the antibiotic. Assuming that 
the mean offspring number is given by $m(c) = 2 / (1 + \alpha c^\beta)$
for some $\alpha, \beta$, where $c$ stands for the 
antibiotic concentration we obtain weakly consistent, 
asymptotically normal estimator both for $(\alpha, \beta)$ and for 
the minimal inhibitory concentration (MIC), a relevant 
parameter in pharmacology. We apply our method to real data, where 
\emph{Chlamydia trachomatis} 
bacteria was treated by azithromycin and ciprofloxacin. 
For the measurements of \emph{Chlamydia} growth
quantitative PCR technique was used.
The 2-parameter model fits remarkably well to 
the biological data. \\


\end{abstract}

\section{Introduction}

Since the discovery of penicillin, antibiotics have been used 
increasingly worldwide to treat bacterial infections.
As the overuse of antibiotics may results drug-resistant bacteria,
determining the bactericidal potency is of the utmost importance.

In the present paper the bacterial population is modeled by a 
Galton--Watson branching process.
The offspring distribution, in particular the offspring mean 
$m(c)$ depends on the antibiotic concentration $c > 0$ as
\begin{equation} \label{eq:m-form}
 m(c) = m_{\alpha, \beta}(c) = \frac{2}{1 + \alpha c^\beta},
\end{equation}
where $\alpha > 0$, $\beta > 0$ are unknown parameters.
Under this model the minimal inhibitory concentration (MIC),
the smallest antibiotic concentration preventing bacterial growth,
is the smallest $c$ for which $m(c) = 1$, that is $\alpha^{-1/\beta}$.
Based on measurements at different concentrations we 
obtain weakly consistent asymptotically normal estimator 
both for $(\alpha, \beta)$, and for the MIC.

We assume that the bacterial population is homogeneous, 
all the cells behave similarly. In particular, 
there is no resistant type. As mutation is rare under normal
conditions and in short time, this is a natural assumption
for our dataset. Long-term evolution of 
bacterial populations with both \emph{resistant} and 
\emph{susceptible} types was investigated in several papers using 
deterministic models,
see Svara and Rankin \cite{Svara}, Paterson et al.~\cite{Paterson}, and 
the references therein.
Closest to our model is the deterministic model given by 
Liu et al.~\cite{Liu}. In \cite{Liu} a deterministic expression 
for the number of colony forming units is obtained in terms 
of the antibiotic concentration. 

Branching processes are classical tools to model cell proliferation, see 
the  monographs by Haccou et al.~\cite{Haccou}, Kimmel and Axelrod 
\cite{Kimmel}. However, to the best of our knowledge 
for estimation of bactericidal potency of 
antibiotics only deterministic models are used.

In the experiments 
growth of \emph{Chlamydia trachomatis} bacterial population was a
analyzed by quantitative PCR (qPCR) method with 12 different antibiotic 
concentrations and 2 different antibiotics.

\emph{Chlamydiae} are obligate intracellular bacteria that primarily infect 
epithelial cells of the conjunctiva, respiratory tract and urogenital
tract. They  
have a unique developmental cycle, with two phenotypic
bacterial forms, the elementary body (EB) and the 
reticulate body (RB). The EB is the infectious form that
can be found outside of the host cells and it is not 
capable to multiply. After infection of the 
host cell, the EB differentiates to RB. 
The RB multiplies in the host cell by binary fission in
a specific area of the infected host  cell, the inclusion. After a 
certain period of time, depending on the chlamydial species,  
the RB redifferentiates to EB. 
The EB is then released from the host cell ready to infect 
new host cells. This unique life-cycle triggered lot of mathematical 
work to model the growth of the population. Wilson \cite{Wilson}
worked out a deterministic model taking into account the infected and 
uninfected host cells and the extracellular \emph{Chlamydia} concentration.
Wan and Enciso \cite{WanEnciso} formulated a deterministic model
for the quantities of RB's and EB's, 
and solved an optimal control problem  
to maximize the quantity of EB's when the host cell dies.
The same problem in a stochastic framework was investigated by 
Enciso et al.~\cite{Enciso} and Lee et al.~\cite{Lee}.
In these papers population growth is modeled without the presence 
of antibiotic.

There is a third form of the bacterium, the aberrant 
body or persistent body. This form is induced by 
various adverse environmental stimuli, such as the lack of nutrients 
and the presence of antibiotics, see Panzetta et al.~\cite{Panzetta}.
The persistent body is not capable to multiply. After elimination of 
the stress stimuli, the persistent body may 
reenter the normal developmental cycle, differentiates to RB,
multiplies and redifferentiates to EB. 
If there is an excess
of antibiotics reaching the 
so-called bactericide concentration, the bacterium is killed, and no 
multiplication can be observed. A lower antibiotic concentration does
not kill all of the bacterium, but leads to 
the formation of non-multiplying aberrant bodies. Further lowering
the antibiotic concentration more RB can be observed,
while the formation of aberrant body 
decreases. At very low antibiotic concentration, the antibiotic 
has no effect on 
the bacterial growth and all the bacteria enter the normal 
developmental cycle. 
Azithromycin and doxycycline are the most commonly
used antibiotics in \emph{Chlamydia} infections (Miller \cite{Miller}), 
but \emph{Chlamydiae} are also sensitive to quinolone type antibiotics 
(Vu et al.~\cite{Vu}).
In our study \emph{Chlamydia trachomatis} infected cells were treated with
azithromycin and the quinolone ciprofloxacin. The dose response curves, 
the concentration dependent impacts of these 
antibiotics on chlamydial growth were measured 48 hours post infection.
A major challenge is the accurate measurement of chlamydial growth.
The golden standard is 
the immunofluorescent labeling and manual counting of the chlamydial 
inclusions, which has several disadvantages, including that the concentration
of the individual bacteria 
cannot be counted. Instead of counting the bacterial cells, the quantity
of bacterial genomes (which is a constant times the number of bacteria) 
can also be measured. Chlamydial
genome concentration in the infected host cells can be measured by a 
quantitative polymerase chain reaction (qPCR). This method is accurate and 
theoretically measures the 
genome of all individual bacteria. 
Eszik et al.~\cite{Virok} developed a version of the qPCR, the 
so-called direct qPCR method for chlamydial growth monitoring. Direct qPCR 
is capable to perform 
qPCR measurements without the labor-intensive DNA purification.
The qPCR method gives a 
so-called cycle threshold (Ct) value to each bacterial sample.
If the effectivity of the qPCR is 100\% then the 
theoretical Ct value equals $a - \log_2 Z^{(i)}_{n; c, x_0}$, where $a \in \R$ is an
unknown constant and $Z_{n; c,x_0}^{(i)}$ stands for the total number 
of dead and alive bacterial cells at antibiotic  concentration $c > 0$,
after $n$ generations starting with $x_0$ bacterias, in experiment $i$.
Adding a measurement error, the measurements have the form 
\begin{equation} \label{eq:C-form1}
C_i(c, x_0) = a - \log_2 Z^{(i)}_{n; c, x_0} + \varepsilon_{i;c},
\qquad i = 1, \ldots, N,
\end{equation}
where  measurement error $\varepsilon_{i;c}$ is assumed to be Gaussian
with mean zero, and variance $\sigma_\varepsilon^2$. This simple linear
model is suggested by Yuan et al.~\cite{Yuan}.
Due to the measurement method lower Ct value means higher genome 
concentration. 
The dose response curves measured by a direct qPCR method are given in 
Figures \ref{fig:simmeas-azi} and \ref{fig:simmeas-cipro}.
\smallskip

The rest of the paper is organized as follows.
The model and some basic properties are given in Section \ref{sect:model}.
The estimator of $m(c)$ for $c$ fixed is provided in Section \ref{sect:m-est},
while in Section \ref{sect:proc} we consider different antibiotic 
concentrations together. Section \ref{sect:sim} contains a small simulation 
study, and real data is analyzed in Section \ref{sect:data}.
The proofs  are gathered together in the Appendix.

\section{The theoretical model} \label{sect:model}

We consider a simple Galton--Watson branching process where the offspring 
distribution depends on the antibiotic concentration $c \geq 0$. 
Each bacteria either dies (leaves no offspring), survives (leaves 1 offspring), 
or divides (leaves 2 offsprings) with respective 
probabilities $p_0 = p_0(c)$, $p_1 = p_1(c)$, and $p_2 = p_2(c)$.
Let $f(s) = f_c(s)$ denote the offspring generating function and 
$m = m(c)$ the offspring 
mean if the antibiotic concentration is $c$, i.e.
\begin{equation*} 
\begin{split}
f(s) &  = f_c(s) = \E s^{\xi_c} = \sum_{i=0}^2 p_{i}(c) s^i,  \quad s \in [0,1],\\
m & = m(c) = f'_c(1) = \E \xi_c,
\end{split}
\end{equation*}
where $\xi_c$ is the number of offsprings.
The process starts 
with $X_0 = x_0$ initial individuals, and
\begin{equation*} 
X_{n+1;c} = \sum_{i=1}^{X_{n;c}} \xi_{i;c}^{(n)},
\end{equation*}
where $\{ \xi_c, \xi_{i;c}^{(n)}: i \geq 1, n \geq 1 \}$ are
independent and identically distributed (iid) random variables 
with generating function $f_c$.
Note that the offspring distribution does depend on the antibiotic
concentration $c$, but here and in the next section 
we suppress this dependence from the notation.

Using the qPCR method the observed quantity is the 
genom of all individual bacteria, which is a constant times the
\emph{total} number 
of bacteria, that is live and dead cells together. Therefore, we have to 
keep track of the dead bacterias too. In order to do this we 
consider a two-type Galton--Watson branching process
$\bX_n = (X_n , Y_n)$, $n \geq 0$, where $X_n$, $Y_n$ stands for the number 
of alive, dead bacterias respectively, in generation $n$. 
Then the total number of bacteria at generation $n$ is $Z_n = X_n + Y_n$.
We also write $Z_{n,x_0}$ to emphasize that $X_0 = x_0$.
The process evolves as
\begin{equation*} 
\begin{split}
X_{n+1 } & = \sum_{i=1}^{X_{n }} \xi_{i }^{(n)} \\
Y_{n+1 } & = Y_{n } + \sum_{i=1}^{X_{n }} \eta_{i }^{(n)}, \quad n \geq 0,
\end{split}
\end{equation*}
$(X_0, Y_0) = (x_0, 0)$,
where 
$(\xi,\eta), (\xi_{i}^{(n)}, \eta_{i}^{(n)})$, 
$n=1,2,\ldots$, $i=1,2,\ldots$ are iid random 
vectors such that
$\p( (\xi ,\eta ) = (0,1) ) = p_0$, 
$\p( (\xi ,\eta ) = (1,0) ) = p_1$,
$\p( (\xi ,\eta ) = (2,0) ) = p_2$.
The offspring mean matrix $\bM$ has the form
\[
\bM =
\begin{pmatrix}
\E \xi & \E \eta \\
0 & 1 
\end{pmatrix}
=
\begin{pmatrix}
m & p_0 \\
0 & 1 
\end{pmatrix}.
\]
Next we determine the mean vector of $\bX_n$.

\begin{lemma} \label{lemma:evar}
If $x_0 = 1$ then for the mean we have
$\E X_n = m^n$, and 
$\E Y_n = p_0 (1 + m + \ldots + m^{n-1})$, thus
\[
\mu_n := \E Z_{n,1} = 
\begin{cases}
m^{n} \left( 1 + \frac{p_0}{m-1} \right) - \frac{p_0}{m-1}, & m \neq 1, \\
1 + p_0 n, & m = 1.
\end{cases}
\]
\end{lemma}

We note that the covariance matrix of $\bX_n$ can be determined
explicitly. The computation is straightforward but rather lengthy.
Since we only need the explicit form of the mean and the finiteness
of the second moments, we skip the computation.

The strong law of large numbers and the central limit theorem
imply that for each fixed $n$ as $x_0 \to \infty$
\begin{equation*} 
\frac{Z_{n,x_0}}{x_0} \longrightarrow \mu_n \quad \text{a.s.}
\end{equation*}
and
\begin{equation} \label{eq:clt}
\frac{Z_{n, x_0} - x_0 \mu_n}{\sqrt{x_0}} 
\stackrel{\mathcal{D}}{\longrightarrow} 
\ N(0, \sigma_n^2),
\end{equation}
where $\stackrel{\mathcal{D}}{\longrightarrow} $ stands for 
convergence in distribution, and 
\[
\sigma_n^2 = \Var (Z_n ).
\]

It is clear that the geometric growth rate  of $\E Z_n$ is 
the offspring mean $m$, 
while the precise distribution determines only the constant 
factor.
Simple analysis shows that if 
$m = p_1 + 2 p_2 > 1$ then
\begin{equation} \label{eq:mu-bound}
m^n \leq \mu_n = \frac{p_2 m^n - p_0}{m-1} \leq 
\frac{m(m^n -1)}{2(m-1)} + 1,
\end{equation}
if $m = 1$ then 
\begin{equation} \label{eq:mu-bound2}
1 \leq \mu_n = 1 + p_0 n \leq 1 + \frac{n}{2},
\end{equation}
while for $m < 1$
\begin{equation} \label{eq:mu-bound3}
1 \leq \mu_n = \frac{p_0 - p_2 m^n }{1-m} \leq 
\frac{m(1-m^n)}{2(1-m)} + 1.
\end{equation}
The upper bound is attained
at $(p_0, p_1, p_2) = (1 - m/2, 0, m/2)$, while 
the lower bound is attained at
$(p_0, p_1, p_2) = (0, 2-m, m-1)$ for $m \geq 1$, and 
at $(p_0, p_1, p_2) = (1-m, m, 0)$ for $m \leq 1$.

The process $(X_n)$ is a single type Galton--Watson process 
with offspring mean $m = p_1 + 2 p_2$.
If $m \leq 1$ then the process dies out 
almost surely, while 
if the process is supercritical, i.e.~$m > 1$ then
the probability of extinction is the smaller 
root of $f(q) = q$, which is
$q = {p_0}/{p_2}$.
By the martingale convergence theorem
\begin{equation} \label{eq:W-lim}
\frac{X_n}{m^n} \to W \quad \text{a.s.},
\end{equation}
where $W$ is a nonnegative random variable. 
For $m\leq 1$ clearly $W \equiv 0$, while
if $m > 1$ then $\p ( W = 0) = q$, 
and the distribution of $W$ is 
absolutely continuous on $(0,\infty)$.

The process $\bX_n = (X_n, Y_n)$ is decomposable, because 
$(\bM^n)_{2,1} = 0$ for any $n$. Limit theorems for 
supercritical decomposable processes
were obtained by Kesten and Stigum \cite{KestenStigum}. The eigenvalues 
of $\bM$ are $m$ and $1$, therefore the process is supercritical if and 
only if $m > 1$. 
Applying Theorem 2.1 by Kesten and Stigum \cite{KestenStigum} we obtain
for $m > 1$ that 
\[
\lim_{n \to \infty} \frac{1}{m^n} ( X_n , Y_n ) = 
W \left( 1 , \frac{p_0}{m-1} \right),
\]
where $W$ is the nonnegative random variable from \eqref{eq:W-lim}.

\section{Estimation of the offspring mean} \label{sect:m-est}

Recall that the measurements are given in the form
\eqref{eq:C-form1}, where $Z_{n;c,x_0}^{(i)}$ stands for the 
total number of dead and alive bacteria at generation $n$, starting 
with $x_0$ bacteria under antibiotic concentration $c$ at experiment 
$i$, $i=1,2,\ldots, N$.
We assume that the sequence $\{ \varepsilon_{i;c} : i \geq 1, c \geq 0 
\}$
are iid, independent of the process $\bX_n$, and is Gaussian with 
mean 0 and variance $\sigma_\varepsilon^2$.

By \eqref{eq:clt}
\[
\begin{split}
\log_2 Z_{n;c,x_0}^{(i)} & = \log_2 ( x_0 \mu_n) + 
\log_2 \left( 1 + \frac{Z_{n;c,x_0}^{(i)} - x_0 \mu_n}{ x_0 \mu_n} \right) \\
& = \log_2 ( x_0 \mu_n) + \frac{1}{\sqrt{x_0} \log 2 } \frac{\sigma_n}{\mu_n} \zeta_i
+ o_\p(x_0^{-1/2}),
\end{split}
\]
where $(\zeta_i)_{i=1,\ldots, N}$ is a sequence of iid $N(0,1)$ random 
variables.
This implies as $x_0 \to \infty$
\begin{equation*} 
C_i(c,x_0) = a -  \log_2 (x_0 \mu_n )
+ \varepsilon_{i;c}
- \frac{1}{\log 2 \, \sqrt{x_0}} \frac{\sigma_n}{\mu_n} \zeta_i
+ o_\p(x_0^{-1/2}).
\end{equation*}
Put
\begin{equation*} 
\log_2 \widehat \mu_n = 
a - \log_2 x_0 - \frac{\sum_{i=1}^N C_i(c,x_0)}{N}.
\end{equation*}
In what follows $\stackrel{\p}{\longrightarrow}$ stands for 
convergence in probability. By the law of large numbers and the 
central limit theorem, we have the following.

\begin{proposition} \label{prop:est}
As $x_0 \to \infty$ and $N \to \infty$
\[
\log_2 \widehat \mu_n 
\stackrel{\p}{\longrightarrow}  \log_2 \mu_n,
\]
which implies that $\widehat \mu_n$ is a weakly consistent
estimator of $\mu_n$.
Furthermore, 
\[
\frac{1}{\sigma_\varepsilon}
\sqrt{N} \left[ 
\log_2 \widehat \mu_n - \log_2 \mu_n
\right] 
\stackrel{\mathcal{D}}{\longrightarrow} N(0,1),
\]
which implies that 
\[
\frac{1}{\sigma_\varepsilon \mu_n \log 2} \,
\sqrt{N} \left( \widehat \mu_n - \mu_n \right) 
\stackrel{\mathcal{D}}{\longrightarrow} N(0,1).
\]
\end{proposition}

We see that from the observations $Z_{n;c,x_0}^{(i)}$ we cannot estimate 
$m$ itself, only $\mu_n$. For $m \in (0,2]$ fixed we obtained
sharp bounds for $\mu_n$ in \eqref{eq:mu-bound}, \eqref{eq:mu-bound2},
\eqref{eq:mu-bound3}. In Figure \ref{fig:mubound} we see the 
corresponding upper and lower bounds for $\log_2 \mu_n$ for $n = 10$.
We see that the larger 
values for $\mu_n$ implies more precise bound for $m$.
Furthermore, larger $n$ also implies more precise bound.
However, for $m \leq 1$ one cannot determine the value $m$. 
This is reasonable, since for both $p_0 = 1$ and 
$p_1 = 1$ we have $\mu_n = 1$, whereas $m= 0$ in the former 
and $m = 1$ in the latter case.

\begin{figure} 
\begin{center}
\includegraphics[width=0.9\textwidth]{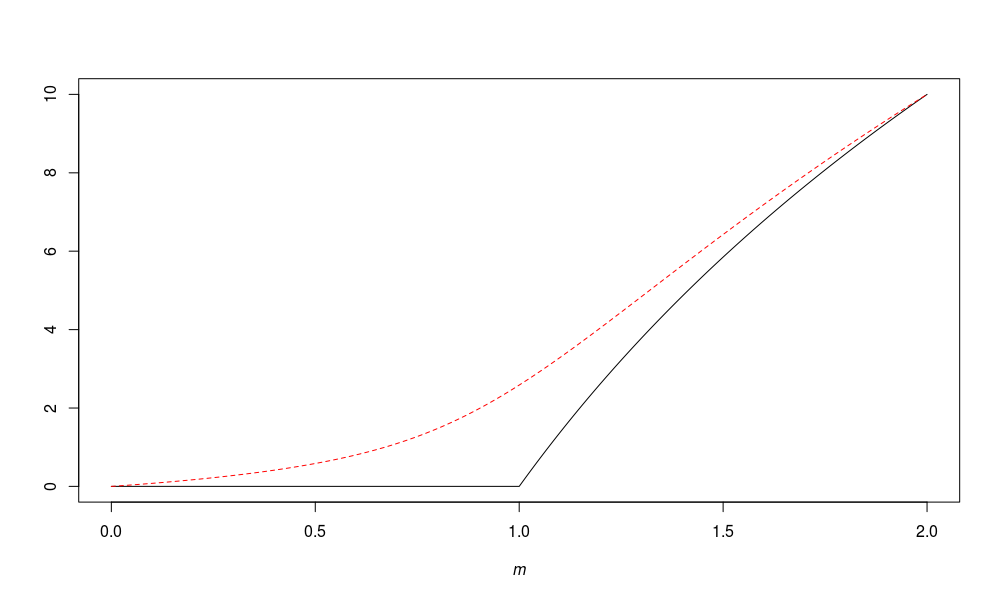}
\end{center}
\vspace*{-.5cm}
\caption{Upper and lower bound for $\log_2 \mu_n$ for $n = 10$.}
\label{fig:mubound}
\end{figure}

However, in real applications we may and do assume that $p_1 \equiv 0$.
This is clearly reasonable for \emph{bactericide} antibiotics, which either 
kill the bacteria, or let it duplicate. While, if a \emph{bacteriostatic} 
antibiotic blocks the duplication of a single bacteria then 
it keeps blocking in the later generations as well. Therefore, we can 
equally count a `blocked' bacteria as a dead one.
Assume now that $p_1 \equiv 0$. Then $\mu_n$ is Lemma \ref{lemma:evar}
simplifies to
\begin{equation*} 
\mu_n ( m) = \frac{m}{2} \left( m^{n-1} + \ldots + 1 \right) + 1
=
\begin{cases}
\frac{m (m^n -1)}{2 (m-1)} + 1, & m \neq 1, \\
\frac{n}{2} + 1, & m = 1.
\end{cases}
\end{equation*}
Then $\mu_n$ is a strictly increasing convex function,
$\mu_n(0) = 1$, $\mu_n(2) = 2^n$. Its inverse function $\psi_n: [1,2^n] \to [0,2]$
is continuous strictly increasing.
Define the estimate
\begin{equation*} 
\widehat m = \psi_n( \widehat \mu_n).
\end{equation*}
From Proposition \ref{prop:est} it follows that 
$\widehat m$ is a weakly consistent estimator of $m$, and
\[
\begin{split}
\psi_n ( \widehat \mu_n ) & = \psi_n \left( \mu_n + 
\frac{\zeta \sigma_\varepsilon \mu_n \log 2}{\sqrt{N}} 
+ o_\p ( N^{-1/2} ) \right) \\
& = \psi_n ( \mu_n) + \psi_n'(\mu_n ) 
\frac{ \zeta  \sigma_\varepsilon \mu_n \log 2}{\sqrt{N}} + o_\p( N^{-1/2}) \\
& = m + \psi_n'(\mu_n ) 
\frac{\zeta \sigma_\varepsilon \mu_n \log 2}{\sqrt{N}} + o_\p( N^{-1/2}),
\end{split}
\]
where $\zeta \sim N(0,1)$.
Noting that $\psi_n'(\mu_n(m)) = 1/\mu_n'(m)$ we obtain the following.

\begin{proposition} \label{prop:m-est}
Assume that $p_1 = 0$.
As $x_0 \to \infty$ and $N \to \infty$, 
$\widehat m$ is a weakly consistent estimator of $m$, and
\begin{equation*} 
\frac{\mu'_n(m)}{ \sigma_\varepsilon \mu_n(m) \log 2} 
\sqrt{N} ( \widehat m - m ) \stackrel{\mathcal{D}}{\longrightarrow}
N(0,1).
\end{equation*}
\end{proposition}

\section{The dependence of $m$ on 
the antibiotic concentration} \label{sect:proc}

Assuming $p_1 \equiv 0$ we can estimate the mean 
for $c > 0$ fixed as 
described in Proposition \ref{prop:m-est}. Next 
we combine our estimator for different concentrations.

We assume that the offspring mean as a function of $c$ 
satisfies \eqref{eq:m-form}
for some unknown parameters $\alpha > 0$, $\beta > 0$. This is 
a quite flexible model, and we show that empirical data fits 
very well to this model. 
Rewriting \eqref{eq:m-form}
\begin{equation*} 
\log \alpha + \beta \log c = \log \left( \frac{2}{m(c)} - 1 \right).
\end{equation*}
Assume that we have measurements for $K \geq 2$ different concentrations
$c_1 < c_2 < \ldots < c_K$, and we obtain the estimator for the offspring 
mean $\widehat m(c_i)$, $ i = 1,2,\ldots , K$. Using simple least squares
estimator we obtain the estimates
\begin{equation} \label{eq:alphahat}
\begin{split}
\widehat \beta & = 
\frac{K \sum_{i=1}^K f_i \ell_i - 
\sum_{i=1}^K f_i L_1 } {K L_2 -  L_1^2 } \\
\widehat \alpha & = 
\exp \left\{ \frac{\sum_{i=1}^K f_i - 
\widehat \beta L_1}{K} \right\},
\end{split}
\end{equation}
where to ease notation we write
\begin{equation*} 
f_i = \log \left( \frac{2}{\widehat m(c_i)} - 1 \right), \quad
\ell_i = \log c_i,
\end{equation*}
and 
\begin{equation} \label{eq:L-def}
L_1 = \sum_{i=1}^K \ell_i, \quad 
L_2 = \sum_{i=1}^K \ell_i^2.
\end{equation}
Note that by the Cauchy--Schwarz inequality the denominator of $\widehat \beta$
is strictly positive for $K \geq 2$.

The \emph{minimal inhibitory concentration} (MIC) is the smallest 
antibiotic concentration that stops bacteria growth. In mathematical 
terms 
\[
\vartheta := \textrm{MIC} = \min \{ c : m(c ) \leq 1 \},
\]
which, under the assumption \eqref{eq:m-form},
$\vartheta = \textrm{MIC} = \alpha^{-1/\beta}$.
Define the estimator 
\begin{equation*} 
 \widehat \vartheta = \widehat \alpha ^{- 1 / \widehat \beta}.
\end{equation*}

In the following statement we summarize the main properties 
of these estimators. Introduce the notation
\begin{equation*} 
k_i = 
- \frac{2}{m(c_i) (2-m(c_i))} \, 
\frac{\sigma_\varepsilon \mu_n(m(c_i)) \log 2}
{\mu_n'(m(c_i))},
\quad i = 1,2,\ldots, K.
\end{equation*}

\begin{proposition} \label{prop:alphabeta}
Assume that $x_0 \to \infty$ and $N \to \infty$. Then 
$\widehat \alpha, \widehat \beta$, and $\widehat \vartheta$
are weakly consistent estimators of the corresponding quantities.
Furthermore,
\[
\sqrt{N}  (\widehat \alpha - \alpha , \widehat \beta - \beta)
\stackrel{\mathcal{D}}{\longrightarrow} (U, V),
\]
where $(U,V)$ is a two-dimensional normal random vector with 
mean 0 and covariance matrix
\[
\begin{pmatrix}
\sigma^2_{\alpha} & \sigma_{\alpha \beta} \\
\sigma_{\alpha \beta} & \sigma^2_{\beta}
\end{pmatrix},
\]
where
\[
\begin{split}
\sigma^2_{\alpha} & = 
\frac{\alpha^2}{\left(K L_2 - L_1^2 \right)^{2}}   \sum_{i=1}^K k_i^2 
(L_2 - L_1 \ell_i)^2 \\
\sigma_{\alpha \beta} & = 
\frac{\alpha}{\left(K L_2 - L_1^2 \right)^{2}} 
\sum_{i=1}^K k_i^2 (K \ell_i - L_1) (L_2 - L_1 \ell_i) \\
\sigma^2_{\beta} & = 
\frac{1}{\left(K L_2 - L_1^2 \right)^{2}} 
\sum_{i=1}^K  k_i^2 \left(K \ell_i - L_1 \right)^2,
\end{split}
\]
and 
\[
\sqrt{N} ( \widehat \vartheta  - \vartheta ) \stackrel{\mathcal{D}}{\longrightarrow}
N(0, \sigma_\vartheta^2),
\]
with
\[
\sigma_\vartheta^2
= 
\frac{\vartheta^2 \, (\log \alpha)^2}{\beta^2 (K L_2 - L_1^2)^2} \sum_{i=1}^K 
k_i^2 \left( 
\frac{L_2 - L_1 \ell_i}{\log \alpha} - \frac{K \ell_i - L_1}{\beta}
\right)^2.
\]
\end{proposition}

\section{Simulation study} \label{sect:sim}

Regardless of the fixed values 
$\bc = (c_1, \ldots, c_K)$ the 
estimator $(\widehat \alpha, \widehat \beta)$ is weakly consistent 
and asymptotically normal as $N \to \infty$. 
However, the asymptotic variances in 
Proposition \ref{prop:alphabeta} do depend on the specific choice 
of $K \geq 2$ and the values $c_1 < \ldots < c_K$. Intuitively, it is clear 
that we should choose values for the concentrations $c_i$ such that 
$m(c_i)$ is not close to 0, nor to 2. 

Consider the following example.
Assume that 
\begin{equation} \label{eq:values}
\alpha = 10, \quad \beta = 1, \quad n= 10, \quad  x_0 = 10^4, \quad
\sigma_\varepsilon = 0.2.
\end{equation}
It turns out that this 
is a reasonable choice, see the azithromycin data in the next section.
The mean offspring function $m(c)$ is given on Figure \ref{fig:mcurve}.

\begin{figure}[h] 
\begin{center}
\includegraphics[width=0.9\textwidth]{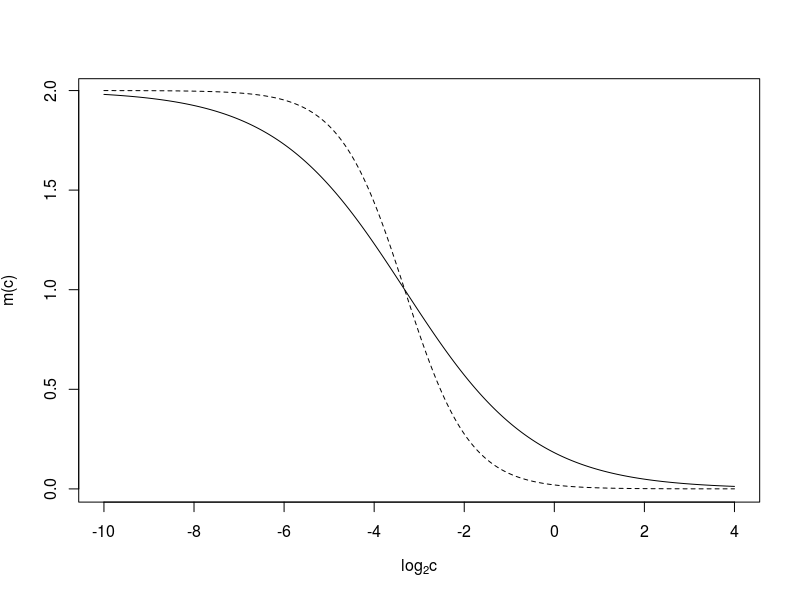}
\end{center}
\vspace*{-.5cm}
\caption{$m(c)$ in a logarithmic scale (solid $(\alpha, \beta) = (10,1)$,
dashed $(\alpha, \beta) = (100,2)$).}
\label{fig:mcurve}
\end{figure}

Choose $K = 3$ different concentrations such that 
$\bc_1 = (2^{-6}, 2^{-4}, 2^{-2})$. Then
for the asymptotic covariances we obtain
\begin{equation} \label{eq:theo-var}
\sigma_\alpha^2 = 8.63, \quad \sigma_{\alpha, \beta} = 0.25, \quad 
\sigma_\beta^2 = 0.00767, \quad \sigma_\vartheta^2 = 0.00012.
\end{equation}
However, as we see in Table \ref{tab:th-cov}
wrong choice of the concentrations might results 
much larger variances. For $\bc_2$
we only observe the process at large concentrations,
killing all the bacteria,
while in case $\bc_3$ the concentration is small, the 
antibiotic does not have any effect.
The combination of large and small values as in $\bc_4$ does not help
either.
Less obvious is the fact that choosing too many points is
contraproductive too. This is the case for $\bc_5$.

\begin{table}[h]
\centering
\begin{tabular}{l | c | c | c | c } 
concentrations & $\sigma^2_\alpha$ &
$\sigma_{\alpha, \beta} $ & $\sigma^2_\beta$ &
$\sigma^2_\vartheta$  \\
\hline \hline 
$\bc_1 = (2^{-6}, 2^{-4}, 2^{-2})$ 
& $8.63$ & $0.25$ &  $ 0.00767 $ & 
$0.00012$ \\ 
$\bc_2 = (2^{-2}, 2^{-1}, 1)$ 
& $112$ & $9.41$ & $0.833$ & $ 0.012$ \\
$\bc_3 = (2^{-9}, 2^{-8}, 2^{-7})$ 
& $967$ & $18.7$ & $0.364$ & $ 0.0298$ \\
$\bc_4= (2^{-8}, 2^{-7}, 2^{-1}, 1)$ 
& $58$ & $1.17$ & $0.0257$ & $ 0.00179$ \\
$\bc_5 = (2^{-9}, 2^{-8},\ldots, 1) $ 
& $23$ & $0.568$ & $0.0157$ & 
$0.00051$ 
\end{tabular} 
\caption{Asymptotic variances for different choices of $\bc$
for $(\alpha, \beta) = (10,1)$.}
 \label{tab:th-cov}
\end{table}

Choosing the values as in \eqref{eq:values}, 
$K =3$ and $\bc_1 = (2^{-6}, 2^{-4}, 2^{-2})$ we simulated the process as
follows. For a given concentration $c_k$, $k=1,\ldots,K$, 
we calculate $m(c_k)$ from 
\eqref{eq:m-form}, and choose the offspring distribution
\[
p_{0;k} = 1 - \frac{m(c_k)}{2}, \quad 
p_{1;k} = 0, \quad 
p_{2;k} = \frac{m(c_k)}{2}.
\]
With this offspring distribution
we simulate $n=10$ generations of the two-type Galton--Watson process
$(X_n, Y_n)$ described in Section \ref{sect:model}. Therefore we 
obtain $Z_{10; c_k, x_0}$. Independently, we repeat the simulation 
$N$ times for each concentration $c_k$. Independent of the $Z$'s
take an iid sequence of Gaussian random variables 
$\{ \varepsilon_{i;c_k}: i=1,\ldots, N; k=1,\ldots,K \}$ with mean
zero and variance $\sigma_\varepsilon^2$. Take $a = 0$ in \eqref{eq:C-form1}. 
The resulting sequence $\{ C_{i}(c_k, x_0): i=1,\ldots, N; k=1,\ldots,K \}$
is one simulated measurement. From each measurement we calculate the 
estimator $(\widehat \alpha, \widehat \beta)$ as described in
\eqref{eq:alphahat}. We simulated the measurement this way 1000 times.
The resulting means and empirical variances of 
$\sqrt{N} ( \widehat \alpha - \alpha, \widehat \beta - \beta)$
and $\sqrt{N} (\widehat \vartheta - \vartheta)$ are given 
in Table \ref{tab:emp-cov}. We see that the empirical values are 
very close to the theoretical ones in 
\eqref{eq:theo-var} even for $N =3,10$.
It is somewhat surprising that the estimates work even for $N=3$,
which is the suggested number of measurements at each 
concentration in microbiology (see e.g.~\cite{Yuan, Virok}).

\begin{table}[h]
\centering
\begin{tabular}{r || c | c | c | c | c | c | c} 
$N$ & $\overline \alpha$ & $\overline \beta$ & $\overline \vartheta$ &
$\widehat \sigma^2_\alpha$ &
$ \widehat \sigma_{\alpha, \beta} $ & 
$\widehat \sigma^2_\beta$ &
$\widehat \sigma^2_\vartheta$  \\
\hline \hline 
$3$  & 10.359 & 1.004 & 0.0998 & 12.95 & 0.325 & 0.00891 & 0.000121 \\
$10$ & 10.106 & 1.002 & 0.1 & 9.27 & 0.262 & 0.00789 & 0.000116 \\
$50$ & 10.03 & 1.0005 & 0.1 & 9.3 & 0.265 & 0.008 & 0.000124 \\
$100$ & 9.999 & 0.9999 & 0.1 & 8.83 & 0.258 & 0.008 & 0.000117 \\
$\infty$ & 10 & 1 & 0.1 & $8.63$ & $0.25$ &  $ 0.00767 $ & $0.00012$ 
\end{tabular} 
\caption{Empirical mean and variances for $(\alpha, \beta) = (10,1)$.}
 \label{tab:emp-cov}
\end{table}

Next we investigate our estimator with a steeper killing 
curve. Let $\alpha = 100$ and $\beta =2$, and the 
other values as in \eqref{eq:values}.  This is also a possible 
choice, see the ciprofloxacin data. In Figure \ref{fig:mcurve}
wee see the mean offspring function $m(c)$ for 
$(\alpha, \beta ) =(10, 1)$ and 
for $(\alpha, \beta) = (100, 2)$. In the latter case there 
are less relevant concentrations, so we expect larger 
variances. In Table \ref{tab:th-cov2} we see that this is partly true,
however the estimate of $\vartheta$ is good.

\begin{table}
\centering
\begin{tabular}{l | c | c | c | c } 
concentrations & $\sigma^2_{\alpha}$ &
$\sigma_{\alpha, \beta} $ & $\sigma^2_{\beta}$ &
$\sigma^2_\vartheta$  \\
\hline \hline 
$\bc_1 = (2^{-6}, 2^{-4}, 2^{-2})$ 
& 11298 & 35.6  &  $ 0.0124 $ & 
$0.000364$ \\ 
$\bc_6 = (2^{-5}, 2^{-4}, 2^{-3})$ 
& 1431 & 5.49 & 0.0216  & 
$0.0000126$ \\ 
$\bc_7= (2^{-7}, 2^{-6}, \ldots , 2^{-1})$ 
& 42490 & 129.3 & 0.429 & $0.00142$ 
\end{tabular} 
\caption{Asymptotics variances for different choices of $\bc$
for $(\alpha, \beta) = (100,2)$.}
 \label{tab:th-cov2}
\end{table}

\section{The experiment} \label{sect:data}

In the experiment $50,000$ mother cells were infected by 
\emph{Chlamydia trachomatis}. The multiplicity of infection (MOI) value, 
the ratio of the initial number of bacteria and number of mother cells is $0.2$.
That is $x_0 = 10, 000$.
The measurements correspond to 12 different antibiotic concentrations using 
two-fold dilution technique, meaning that $c_i = 2^i c_0$, $i=0,1,\ldots, 11$.
For each concentration 3 measurements were done.
For the technical details of the experiment we refer to
\cite{Virok}.

We analyze two antibiotics: azithromycin and ciprofloxacin.
These antibiotics have different antimicrobial effects:
azithromycin is a bacteriostatic antibiotic, meaning that it 
does not necessarily kill the bacteria, only prevents growth, while
ciprofloxacin is a bactericide antibiotic, which usually kills bacteria.
In Figures \ref{fig:simmeas-azi} and \ref{fig:simmeas-cipro} we see the qPCR 
measurements as a function of $\log_2 c$.

\begin{figure} 
\begin{center}
\includegraphics[width=0.9\textwidth]{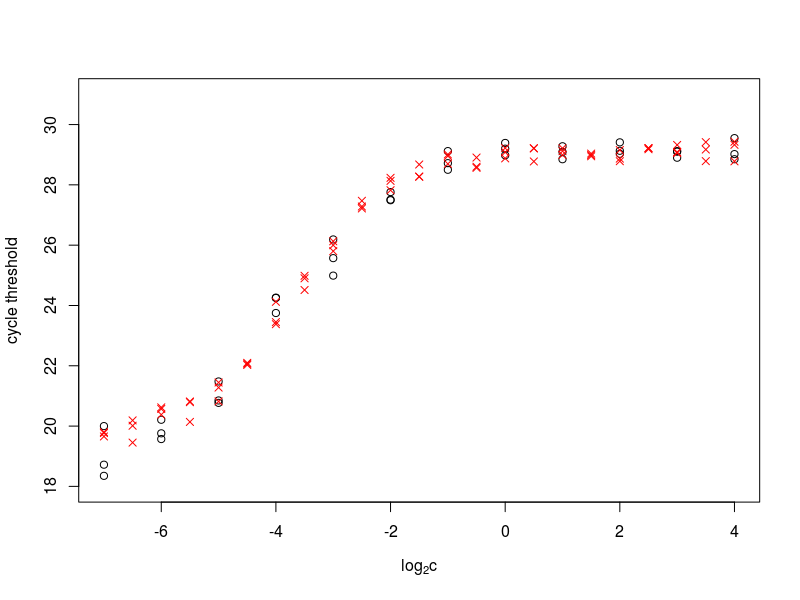}
\end{center}
\vspace*{-.5cm}
\caption{Measured ($\circ$) and simulated $(\times)$ Ct values for 
azithromycin.}
\label{fig:simmeas-azi}
\end{figure}

\begin{figure} 
\begin{center}
\includegraphics[width=0.9\textwidth]{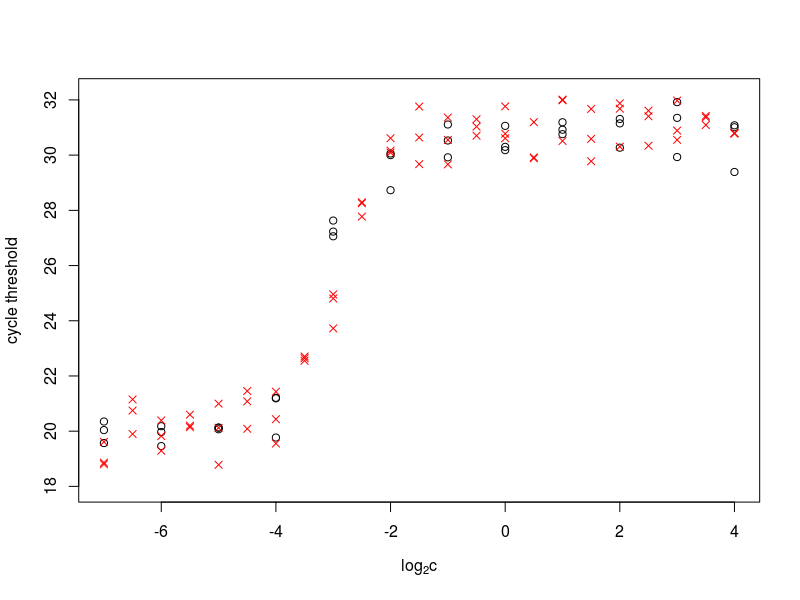}
\end{center}
\vspace*{-.5cm}
\caption{Measured ($\circ$) and simulated $(\times)$ Ct values for 
ciprofloxacin.}
\label{fig:simmeas-cipro}
\end{figure}

If $c$ is large enough, i.e.~at very high antibiotic 
concentration $m(c)$ is close to 0, that is
$Z_{n;x_0, c} \approx x_0$, since all the 
bacteria dies without offspring. Therefore, for $c$ large enough
we can estimate the constant $a$ in \eqref{eq:C-form1} as
\[
\widehat a_N = \frac{1}{N} \sum_{i=1}^{N} C_i(c, x_0) + \log_2 x_0.
\]
Then $\widehat a_N$ is normally distributed with mean $a$ and 
variance $\sigma_\varepsilon^2/N$. Furthermore, $\sigma_\varepsilon$ can also be 
estimated from these data. For azithromycin we used the 
measurements where $c \geq 2^{-1}$, while for ciprofloxacin
$c \geq 1$. 

The number of generations $n$ is typically a fixed small number,
in our experiments around 10. 
If $c $ is small then there is no antibiotical effect so the bacterial
population grows freely, that is $Z_{n, x_0, 0} \approx 2^n x_0$. 
We can estimate $n$ as
\[
\widehat n_N = \widehat a_N - \log_2 x_0 - \frac{1}{N} \sum_{i=1}^N C_i(c, x_0). 
\]
Then $\widehat n_N$ is normally distributed with mean $n$ and
variance $2 \sigma_\varepsilon^2/N$. To estimate $\widehat n_N$ we used 
the smallest possible concentration, $c = 2^{-7}$.

Using Proposition \ref{prop:m-est} we estimate $m(c)$.
In Figures \ref{fig:azi-fitted} and \ref{fig:cipro-fitted}
we see the estimated means and the corresponding fitted 
curve $m(c)$, where the parameters $\alpha, \beta$ are 
estimated as described in \eqref{eq:alphahat}.
In the previous section we showed that the best strategy is 
to choose few concentration where the mean offspring is not 
close to 0, nor to 2. For the azithromycin we chose 
$\bc = (2^{-5}, 2^{-4}, 2^{-2}, 2^{-1})$ and obtained 
$\widehat \alpha = 9.1$, $\widehat \beta = 1.12$, and 
$\widehat \vartheta = 0.139$. (We obtain similar estimates 
for various reasonable choices.) 
For ciprofloxacin in Figure \ref{fig:simmeas-cipro} we 
see a rapid drastic change; for $c \geq 2^{-2}$ the population 
dies out, while for $c \leq 2^{-4}$ the population freely grows. We chose 
$\bc = (2^{-4}, 2^{-3}, 2^{-2})$ and obtained 
$\widehat \alpha = 71.8$, $\widehat \beta = 2.46$, 
$\widehat \vartheta = 0.175$. (These values are less stable 
to the change in $\bc$.)
Simulated measurements with the estimated  values are given in 
Figures \ref{fig:simmeas-azi} and \ref{fig:simmeas-cipro}, where 
the circles are the real measurements and the crosses are the 
simulated ones. In both cases we obtain a remarkably good fit.

\begin{figure} 
\begin{center}
\includegraphics[width=0.9\textwidth]{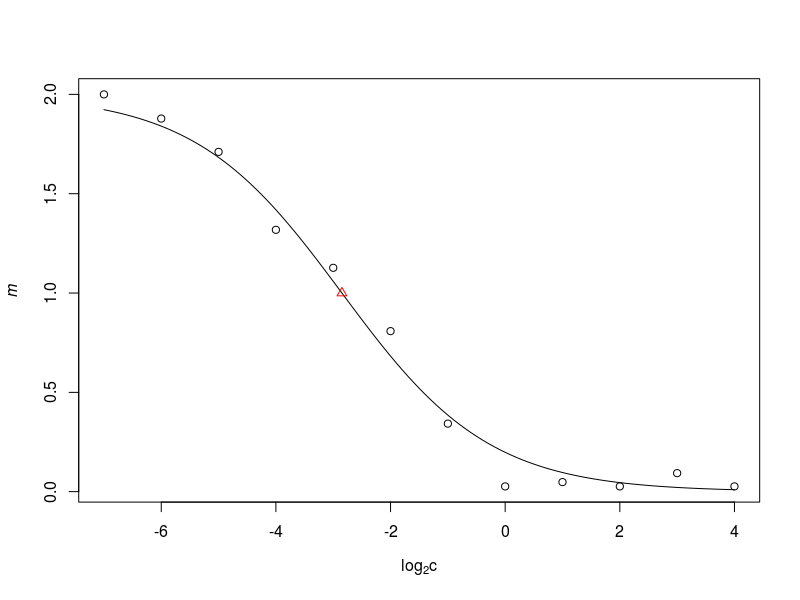}
\end{center}
\vspace*{-.5cm}
\caption{Estimated means and the fitted curve for azithromycin.}
\label{fig:azi-fitted}
\end{figure}

\begin{figure}
\begin{center}
\includegraphics[width=0.9\textwidth]{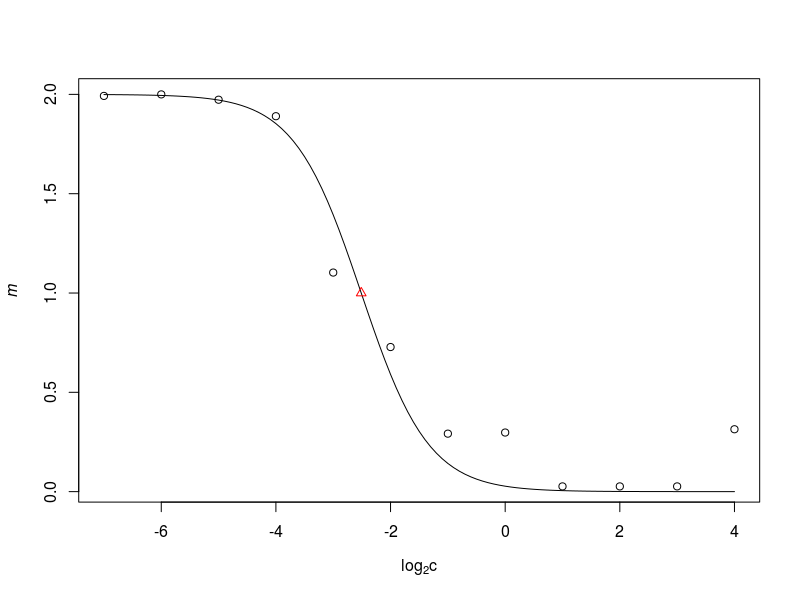}
\end{center}
\vspace*{-.5cm}
\caption{Estimated means and the fitted curve for ciprofloxacin.}
 \label{fig:cipro-fitted}
\end{figure}

\section{Conclusion}

To model the growth of a bacterial population
and its dependence on the antibiotic concentration 
we proposed a simple Galton--Watson model, where 
the offspring distribution depends on the antibiotic 
concentration via \eqref{eq:m-form}. A stochastic model
is more natural compared to 
the previous deterministic model in \cite{Liu}, because
we are able to estimate the parameters of the model and 
investigate the properties of the estimator. 
Taking into account the measurement error using qPCR technique,
from the measurements at different antibiotic concentrations
we obtained 
a weakly consistent asymptotically normal estimator for 
the unknown parameters $(\alpha, \beta)$ in \eqref{eq:m-form}.

The \emph{minimal inhibitory concentration} (MIC), the smallest 
concentration of antibiotic that prevents bacterial growth, is a
very important parameter in pharmacology. Its estimation is 
rather troublesome, since due to the usual 
two-fold dilution technique one can observe only the bacterial
growth under antibiotic concentration  
$c_0, 2c_0, \ldots, 2^k c_0$. Therefore one can claim only that 
the MIC belongs to some interval $[c,2c]$, or give an upper bound 
for it. The vast majority of the literature does not provide a proper 
mathematical model for the growth of the bacterial population, 
only determines the MIC value as the smallest concentration 
without visible bacterial growth. In our framework an explicit
mathematical definition of the MIC is given, and we constructed 
an estimator for it.

Simulation study showed that the estimators work surprisingly
well even if the number of measurements at different concentration
is 3, which is the suggested number in microbiology 
(see e.g.~\cite{Yuan, Virok}).

We applied the model to real measurements, where growth
of \emph{Chlamydia trachomatis} was analyzed treated by two 
different antibiotics.  
Although the mathematical model has only 2 parameters,
we found extremely good fitting to the real data for 
both the bactericide and the bacteriostatic antibiotic.

\section*{Appendix}

\begin{proof}[Proof of Lemma \ref{lemma:evar}]
Conditioning on $\bX_n$
\[
\E \left[ \bX_{n+1} | \bX_n \right]
=
\begin{pmatrix}
m X_n \\
p_0 X_n + Y_n
\end{pmatrix}
= \bX_n \bM,
\]
thus
\[
\E \bX_n = \bX_0 \bM^n.
\]
We have, by induction on $n$ that
\[
\bM^n =
\begin{pmatrix}
m^n & p_0 (1 + \ldots + m^{n-1}) \\
0 & 1
\end{pmatrix},
\]
thus 
\begin{equation*} 
\begin{split}
\E Z_n & = m^n + p_0 ( 1 + m + \ldots + m^{n-1} ) \\
& =
\begin{cases}
m^n \left( 1 + \frac{p_0 }{m-1} \right) - \frac{p_0 }{m-1}, 
& \text{if } m \neq 1, \\
 1 + n p_0 , & \text{if } m = 1,
\end{cases}
\end{split}
\end{equation*}
as claimed. 
\end{proof}

\begin{proof}[Proof of Proposition \ref{prop:alphabeta}]
Again by Proposition \ref{prop:m-est}
\[
\begin{split}
f_i & =
\log \left( \frac{2}{m(c_i)} - 1 \right)  
- \frac{2}{m(c_i) (2-m(c_i))} \, \frac{\sigma_\varepsilon \mu_n(m(c_i)) \log 2}
{\sqrt{N} \mu_n'(m(c_i))} \zeta_i  + o_\p(N^{-1/2}) \\
 & = \log \left( \frac{2}{m(c_i)} - 1 \right) + \frac{k_i}{\sqrt{N}} 
\zeta_i
+ o_\p(N^{-1/2}),
\end{split}
\]
where $\zeta_i$'s are iid $N(0,1)$, $i=1,2,\ldots, K$. Recall 
the notation in \eqref{eq:L-def}. Then
\[
\begin{split}
\sum_{i=1}^K f_i \left( K \ell_i  - L_1 \right)
=\sum_{i=1}^K 
\left[ \log \left( \frac{2}{m(c_i)} - 1 \right) + \frac{k_i \zeta_i}{\sqrt{N}}
\right]
\left( K \ell_i - L_1 \right)
+ o_\p(N^{-1/2}).
\end{split}
\]
Substituting back into \eqref{eq:alphahat}
\begin{equation} \label{eq:betahat-asy}
\widehat \beta - \beta = 
\frac{1}{\sqrt{N}} 
\sum_{i=1}^K  \zeta_i \, k_i 
\frac{ K \ell_i - L_1}{K L_2 - L_1^2}
+ o_\p(N^{-1/2}), 
\end{equation}
and similarly
\begin{equation} \label{eq:alphahat-asy}
\begin{split}
\log \widehat \alpha - \log \alpha & = \frac{1}{K \sqrt{N}}
\sum_{i=1}^K \zeta_i \, k_i \left( 1 - 
\frac{L_1 (K \ell_i - L_1) }{K L_2 - L_1^2 }
\right) + o_\p(N^{-1/2}) \\
& = \frac{1}{\sqrt{N}}
\sum_{i=1}^K \zeta_i \, k_i 
\frac{L_2 - L_1 \ell_i}{K L_2 - L_1^2 }
+ o_\p(N^{-1/2}),
\end{split}
\end{equation}
which implies
\[
\widehat \alpha - \alpha = \frac{\alpha}{\sqrt{N}}
\sum_{i=1}^K \zeta_i \, k_i 
\frac{ L_2 - L_1 \ell_i   }{K L_2 - L_1^2 }
+ o_\p(N^{-1/2}). 
\]
From \eqref{eq:betahat-asy} and \eqref{eq:alphahat-asy} 
we obtain
\[
\widehat \vartheta - \vartheta = 
- \frac{\vartheta \log \alpha}{\sqrt{N} \beta  (K L_2 - L_1^2)}
\sum_{i=1}^K \zeta_i k_i \left( 
\frac{L_2 - L_1 \ell_i}{\log \alpha} - \frac{K \ell_i - L_1}{\beta}
\right).
\]
\end{proof}

\noindent \textbf{Acknowledgements.}
Kevei is supported 
by the J\'anos Bolyai Research Scholarship of the Hungarian
Academy of Sciences,  
by the NKFIH grant FK124141, and 
by the EU-funded Hungarian grant EFOP-3.6.1-16-2016-00008.
Szalai's research was partially supported by the
the EU-funded Hungarian grant EFOP-3.6.2-16-2017-00015 2020,
and by the 
Ministry of Human Capacities, 
Hungary grant TUDFO/47138-1/2019-ITM.


\end{document}